\documentclass[a4paper,UKenglish]{lipics-v2018}

\usepackage{tikz, stmaryrd}
\usepackage{tikz-cd}

\usetikzlibrary{matrix}
\usepackage{graphicx}
\usepackage{pgfplotstable}
\usepackage{amssymb}
\usepackage{mathtools}
\usepackage[utf8]{inputenc}
\usepackage{amsmath} 
\usepackage{amsthm}
\usepackage{algorithm}
\usepackage{algpseudocode}
\usepackage{hyperref}

\usepackage{xcolor,colortbl}

\definecolor{Gray}{gray}{0.85}
\definecolor{LightCyan}{rgb}{0.88,1,1}

\newcolumntype{a}{>{\columncolor{Gray}}c}
\newcolumntype{b}{>{\columncolor{white}}c}
 
\usepackage{microtype}

\title{Strong Collapse for Persistence\footnote{This research has received funding from the
European Research Council (ERC) under the European Union’s Seventh Framework Programme (FP/2007-
2013) / ERC Grant Agreement No. 339025 GUDHI (Algorithmic Foundations of Geometry Understanding in
Higher Dimensions).}}
\titlerunning{Strong Collapse for Persistence} 

\author{Jean-Daniel Boissonnat}{Universit\'e C\^ote d'Azur, INRIA, France}{Jean-Daniel.Boissonnat@inria.fr}{}{}
\author{Siddharth Pritam}{Universit\'e C\^ote d'Azur, INRIA, France}{siddharth.pritam@inria.fr}{}{}
\author{Divyansh Pareek}{Indian Institute of Technology Bombay, India }{divyansh@cse.iitb.ac.in}{}{}

\authorrunning{J-D.\,Boissonnat and S. Pritam and D. Pareek } 
\Copyright{Jean-Daniel Boissonnat, Siddharth Pritam, and Divyansh Pareek}
\subjclass{Mathematics of computing, Topological Data Analysis, Computational geometry}

\nolinenumbers
\newcommand{\red}[1]{{\color{red}#1\normalcolor}}
\newcommand{\blue}[1]{{\color{blue}#1\normalcolor}}

\begin{document}

\maketitle

\begin{abstract}
We introduce  a fast and memory efficient approach to compute the persistent homology (PH) of a sequence of simplicial complexes.  The basic idea is to  simplify the complexes of the input sequence by using strong collapses, as introduced by J. Barmak and E. Miniam [DCG (2012)], and to compute the PH of an induced sequence of reduced simplicial complexes that has the same PH as the initial one.
Our approach has several salient features that distinguishes it from previous work.  It is not limited to  filtrations (i.e. sequences of nested  simplicial subcomplexes) but works for other types of sequences like towers and zigzags. To  strong collapse a simplicial complex, we only need to store the maximal simplices of the complex, not the full set of all its simplices, which saves a lot of space and time. Moreover,
the complexes in the sequence can be strong collapsed independently and in parallel. Finally, we can compromize between precision and time by choosing the number of simplicial complexes of the sequence we strong collapse.
As a result and as demonstrated by numerous experiments on publicly available data sets, our approach is extremely fast and memory efficient in practice.

\end{abstract}

\section{Introduction}
In this article, we address the problem of computing the Persistent Homology (PH) of a given sequence of simplicial complexes (defined precisely in Section 4) in an efficient way. It is known that computing persistence can be done in $O(n^\omega)$ time, where $n$ is the total number of simplices and $\omega \leq 2.4$ is the matrix multiplication exponent ~\cite{ZigzagMatrixTime, MatrixMult}.  In practice,  when dealing with massive and high-dimensional datasets, $n$ can be very large (of order of billions) and computing PH is then  very slow and memory intensive. Improving the performance of PH computation  has therefore become an important research topic in Computational Topology and Topological Data Analysis. 

Much progress has been accomplished  in the recent years in two directions. First, a number of clever implementations and  optimizations have led to  a new generation of software for PH computation~\cite{Gudhi,phat,ripser,dionysus}. Secondly, a complementary direction has been explored to reduce the size of the complexes in the sequence while preserving (or approximating in a controlled way) the persistent homology of the sequence. Examples are the work of Mischaikow and Nanda ~\cite{Mischaikow} who use Morse theory to reduce the size of a filtration, and the work of Dłotko and Wagner  who use simple collapses~\cite{PawelSimpleColl}. Both methods compute the exact PH of the input sequence. Approximations can also be computed with theoretical guarantees. Approaches like interleaving smaller and easily computable simplicial complexes, and sub-sampling of the point sample works well upto certain approximation factor ~\cite{ChazalOudot,Botnan,SheehyRipsComp, KerberSharath, Aruni, Simba}. 

In this paper, we introduce a new approach to  simplify the complexes of the input sequence which uses the notion of strong collapse introduced by J. Barmak and E. Miniam \cite{StrongHomotopy}.  
Specifically, our approach 
can be summarized as follows. Given a sequence $\mathcal{Z}$ : $\{K_1 \xrightarrow{f_1} K_2 \xleftarrow{g_2} K_3 \xrightarrow{f_3} \cdots \xrightarrow{f_{(n-1)}} K_n\}$ of simplicial complexes $K_i$ connected through simplicial maps $\{ \xrightarrow{f_i}$ or $ \xleftarrow{g_j} \}$,  we independently strong collapse the complexes of the sequence to reach a sequence $\mathcal{Z}^c$ : $\{K_1^c \xrightarrow{f_1^c} K_2^c \xleftarrow{g_2^c} K_3^c \xrightarrow{f_3^c} \cdots \xrightarrow{f_{(n-1)}^c} K_n^c\}$, with \textit{induced simplicial maps}  $\{ \xrightarrow{f_i^c}$ or $ \xleftarrow{g_j^c} \}$ (defined in Section 4). The complex $K_i^c$ is called the \textbf{core} of the complex $K_i$ and we call the sequence $\mathcal{Z}^c$ the \textbf{core sequence} of $\mathcal{Z}$. We show that one can compute the PH of the sequence $\mathcal{Z}$ by computing the PH of the core sequence $\mathcal{Z}^c$, which is of much smaller size. 

Our method has some similarity with the work of Wilkerson et. al.~\cite{StrongCollPerst} who also use strong collapses to reduce PH computation  but it differs in three essential aspects: it is not limited to  filtrations (i.e. sequences of nested  simplicial subcomplexes) but works for other types of sequences like towers and zigzags.  It also differs in the way strong collapses are computed and in the manner PH is computed. 

A first central observation is that to strong collapse a simplicial complex $K$, we only need to store its maximal simplices (i.e. those simplices that have no coface).  The number of maximal simplices is smaller than the total number of simplices by a factor that is exponential in the dimension of the complex. It is linear in the number of vertices for a variety of complexes~\cite{gamma2}. Working only with maximal simplices dramatically reduces the time and space complexities compared to the algorithm of~\cite{StrongCollKrim}. We prove that the complexity of our algorithm is $\mathcal{O}(v^2\Gamma_0d+ m^2\Gamma_0d)$. Here $d$ is the dimension of the complex, $v$ is the number of vertices, $m$ is the number of maximal simplices and $\Gamma_0$ is an upper bound on the number of maximal simplices incident to a vertex. As observed in~\cite{gamma,gamma2}, usually $m$ is  much smaller than the total number of simplices and $\Gamma_0$ is much smaller than $m$ (see Section~\ref{sec:algo} for a discussion).

We now consider PH computation.  All PH algorithms take as input a full representation of the complexes. 
We thus have to convert the representation by maximal simplices used for  strong collapses into a full representation of the complexes, which takes exponential time in the dimension (of the collapsed complexes). This exponential burden is to be expected since it is 
 known that computing PH is NP-hard when the complexes are represented by their maximal faces~\cite{NpHardHomology}. Nevertheless, we demonstrate in this paper that strong collapses combined with known persistence algorithms lead to major improvements over previous methods to compute the PH of a sequence. This is due in part to the fact that strong collapses reduce the size of the complexes on which persistence is computed.   Two other factors also play a role:

-- The collapses of the complexes in the sequence can be performed independently and in parallel. This is due to the fact that strong collapses can be expressed as simplicial maps unlike simple collapses \cite{Whitehead}.

-- The size of the complexes in a sequence  does not grow by much in terms of maximal simplices, as observed in many practical cases.  
As a consequence, the time to collapse the $i$-th simplicial complex $K_i$ in the sequence is almost independent of $i$.
For filtrations, this is a clear advantage over methods that use a full representation of the complexes and suffer an increasing cost as $i$ increases. 

-- Using our approach, one can compute the exact PD or a certified approximation by strong collapsing only  a subset of the simplicial complexes of the sequence.  We can thus compromize between precision and time.

As a result,  our approach is extremely fast and memory efficient in practice as demonstrated by numerous experiments on publicly available data sets.

An outline of this paper is as follows. Section 2 recalls the basic ideas and constructions related to simplicial complexes and strong collapses. We describe our core algorithm in Section 3. In Section 4, we prove that zigzag modules are preserved  under strong collapse. In Section 5, we provide experimental results and we end with a short discussion in Section 6.

\section{Preliminaries}
In this section,  we provide a brief review of the notions of simplicial complex and strong collapse as introduced in~\cite{StrongHomotopy}. We assume some familiarity with basic concepts like homotopic maps, homotopy type, homology groups and other algebraic topological notions. Readers can refer to \cite{Hatcher} for a comprehensive introduction of these topics.
\subparagraph{Simplex, simplicial complex and simplicial map :}

An \textbf{abstract simplicial complex} $\textit{K}$ is a collection of subsets of a non-empty finite set $\textit{X},$ such that for every subset $\textit{A}$ in $\textit{K}$, all the subsets of $\textit{A}$ are in $\textit{K}$. From now on we will call an \textit{abstract simplicial complex} simply a \textit{simplicial complex} or just a \textit{complex}. An element of $\textit{K}$ is called a \textbf{simplex}. An element of cardinality $k+1$ is called a $k$-simplex and 
$k$ is called its \textbf{dimension}. A simplex is called \textbf{maximal} if it is not a proper subset of any other simplex in $\textit{K}$. A sub-collection  $\textit{L}$ of $\textit{K}$ is called a \textbf{subcomplex}, if it is a simplicial complex itself. $\textit{L}$ is a \textbf{full subcomplex} if it contains all the simplices of $\textit{K}$ that are spanned by the vertices (0-simplices) of the subcomplex $\textit{L}$.

A vertex to vertex map $\psi : K \rightarrow L$ between two simplicial complexes is called a \textbf{simplicial map}, if the images of the vertices of a simplex always span a simplex. Simplicial maps are thus determined by the images of the vertices. In particular, there is a finite number of simplicial maps between two given finite simplicial complexes. Simplicial maps induce continuous maps between the underlying \textit{geometric realisations} of the simplicial complexes. {Two simplicial maps $\phi : K \rightarrow L$ and $\psi : K \rightarrow L$ are contiguous if, for all $\sigma \in K$, $\phi(\sigma) \cup \psi(\sigma) \in L$. Two contiguous maps are known to be homotopic ~\cite[Theorem 12.5]{Munkres}}.

\subparagraph{Dominated vertex: }
Let $\sigma$ be a simplex of a simplicial complex $K$, the \textbf{closed star} of $\sigma$ in $K$, $st_K(\sigma)$ is a subcomplex of $K$ which is defined as follows,
$ st_K(\sigma) := \{ \tau \in K | \hspace{5px} \tau \cup \sigma \in K \}.$
The \textbf{link} of $\sigma$ in $K$, $lk_K(\sigma)$ is defined as the set of simplices in $st_K(\sigma)$ which do not intersect with $\sigma$,
$ lk_K(\sigma) := \{ \tau \in st_K(\sigma) | \tau \cap \sigma = \emptyset \}.$

Taking a join with a vertex transforms a simplicial complex into a simplicial cone. Formally if $L$ is a simplicial complex and $a$ is a vertex not in $L$ then the \textbf{simplicial cone} $aL$ is defined as  $aL := \{ a,\hspace{2px} \tau \hspace{2px}|\hspace{5px} \tau \in L \hspace{5px} or \hspace{5px} \tau = \sigma \cup a; \hspace{5px} where \hspace{5px} \sigma \in L \}$. A vertex $v$ in $K$ is called a \textbf{dominated vertex} if the link of $v$ in $K$, $lk_K(v)$ is a simplicial cone, that is, there exists a vertex $v^{\prime} \neq v$ and a subcomplex $L$ in $K$, such that $lk_K(v) = v^{\prime}L$. We say that the vertex  $v^{\prime}$ is \textit{dominating} $v$ and $v$ is \textit{dominated} by $v^{\prime}$. The symbol \textbf{$K \setminus v$} (deletion of $v$ from $K$) refers to the subcomplex of $K$ which has all simplices of $K$ except the ones containing $v$. Below is an important remark from ~\cite[Remark 2.2]{StrongHomotopy}, which proposes an alternative definition of dominated vertices.

\label{Rem:equDef}\textbf{Remark 1:} A vertex $v \in K$ is dominated by another vertex $v^{\prime} \in K$, \textit{if and only if} all the maximal simplices of $K$ that contain $v$ also contain $v^{\prime}$ \cite{StrongHomotopy}.

\subparagraph{Strong collapse: } An \textbf{elementary strong collapse} is the deletion of a dominated vertex $v$ from $K$, which we denote with $K$ ${\searrow\searrow}^e$ $K\setminus v$. \autoref{fig:elementrycollapse} illustrates an easy case of an elementary strong collapse. There is a \textbf{strong collapse} from a simplicial complex $K$ to its subcomplex $L$, if there exists a series of elementary strong collapses from $K$ to $L$, denoted as $K$ ${\searrow\searrow}$ $L$. The inverse of a strong collapse is called a \textbf{strong expansion}. If there exists a combination of strong collapses and/or strong expansion from $K$ to $L$ then $K$ and $L$ are said to have the same \textbf{strong homotopy type}.

The notion of strong homotopy type is stronger than the notion of simple homotopy type in the sense that if $K$ and $L$ have the same strong homotopy type, then they have the same simple homotopy type, and therefore the same homotopy type \cite{StrongHomotopy}. There are examples of contractible or simply collapsible simplicial complexes that are not strong collapsible.

 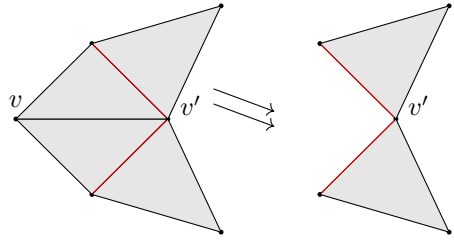
\begin{figure}
 \begin{center}
 \begin{tikzpicture}
 [scale=2, vertices/.style={draw, fill=black, circle, inner sep=0.5pt}]
 \node[vertices] (a) at (-0.5,-0.5) {};
 \node[vertices, label =$v'$] (b) at (0,0) {};
 \node[vertices] (c) at (-0.5,0.5) {};
 \node[vertices, label =$v$] (d) at (-1.0,0.0) {};
 \node[vertices] (e) at (0.35,0.75){};
 \node[vertices] (f) at (0.35,-0.75){};
 \node[ label =$v'$] (p) at (0.15,-0.1) {};
 
 \filldraw[fill=black!10, draw=black] (-0.5,-0.5)--(0,0)--(-1.0,0.0)--cycle ;
 \filldraw[fill=black!10, draw=black] (-0.5,0.5)--(0,0)--(-1.0,0.0)--cycle;
 \filldraw[fill=black!10, draw=black] (0,0)--(-0.5,0.5)--(0.35,0.75)--cycle;
 \filldraw[fill=black!10, draw=black] (0,0)--(-0.5,-0.5)--(0.35,-0.75)--cycle;
 \draw[draw =red] (a)--(b);
 \draw[draw =red] (c)--(b);
 \draw [->] (0.3,0.2) -- (0.7,0.05); 
 \draw [->] (0.3,0.10) -- (0.7,-0.05);

 \node[vertices] (k) at (1.0, -0.5) {};
 \node[vertices, label =$v^\prime$] (l) at (1.5, 0) {};
 \node[vertices] (m) at (1.0, 0.5) {};
 \node[vertices] (n) at (1.85,0.75){};
 \node[vertices] (o) at (1.85,-0.75){};
 \node[ label =$v'$] (q) at (1.65,-0.1) {};
 
 \filldraw[fill=black!10, draw=black] (1.5,0)--(1.0,-0.5)--(1.85,-0.75)--cycle;
 \filldraw[fill=black!10, draw=black] (1.5,0)--(1.0,0.5)--(1.85,0.75)--cycle;
 \draw [draw =red](k)--(l);
 \draw [draw =red](l)--(m);
 \end{tikzpicture}
 \end{center}
 \caption{Illustration of an \textit{elementary strong collapse}. In the complex on the left, $v$ is dominated by $v'$. The link of $v$ is highlighted in {\color{red} red}. Removing $v$ leads to the complex on the right.}
 \label{fig:elementrycollapse}
 \end{figure}

A complex without any dominated vertex will be called a \textbf{minimal complex}. A \textbf{core} of a complex $K$ is a minimal subcomplex $K^c \subseteq K$, such that $K$ ${\searrow\searrow}$ $K^c$. \textit{Every simplicial complex has a \textbf{unique core} up to isomorphism. The core decides the strong homotopy type of the complex}, and two simplicial complexes have the same strong homotopy type \textit{if and only if} they have isomorphic cores~\cite[Theorem 2.11]{StrongHomotopy}.

\subparagraph{Retraction map:} If a vertex $v \in K$ is dominated by another vertex $v^{\prime} \in K$, the vertex map $r : K \rightarrow K \setminus v$ defined as: $r(w) = w$ if $w \neq v$ and $r(v) = v^{\prime}$, induces a simplical map that is a \textit{retraction} map. The homotopy between $r$ and the identity $i_{K \setminus v}$ over $K \setminus v$ is in fact a strong deformation retract. {Furthermore, the composition $(i_{K \setminus v})r$ is contiguous to the identity $i_{K}$ over $K$~\cite[Proposition 2.9]{StrongHomotopy}}.  \label{map:retmap}

\subparagraph{Nerve of a simplicial complex: } \label{sec:Nerve}
A closed \textbf{cover} $\mathcal{U}$ of a topological space $\mathcal{X}$ is a set of closed sets of $\mathcal{X}$ such that $\mathcal{X}$ is a subset of their union. The \textbf{nerve} of a cover $\mathcal{U}$ is an abstract simplicial complex, defined as the set of all non-empty intersections of the elements of $\mathcal{U}$. The nerve is a well known construction that transforms a continuous space to a combinatorial space preserving its homotopy type. The \textit{nerve} $\mathcal{N}(K)$ of a simplicial complex $K$ is defined as the nerve of the set of maximal simplices of the complex $K$ (considered as a cover of the complex). Hence all the maximal simplices of $K$ will be the vertices of $\mathcal{N}(K)$ and their non-empty intersection will form the simplices of $\mathcal{N}(K)$. For $j \geq 2$ the iterative construction is defined as $\mathcal{N}^j(K) =\mathcal{N}(\mathcal{N}^{j-1}(K))$. This definition of nerve preserves the homotopy type, $K \simeq \mathcal{N}(K)$\cite{StrongHomotopy}. A remarkable property of this nerve construction is its connection with strong collapses.

Taking the nerve of any simplicial complex $K$ twice corresponds to a strong collapse.

\begin{theorem}~\cite[Proposition 3.4]{StrongHomotopy} \label{Thm:NerveSquare}
For a simplicial complex $K$, there exists a subcomplex $L$ isomorphic to $\mathcal{N}^2(K)$, such that $K {\searrow\searrow} L$.
\end{theorem} 

An easy consequence of this theorem is that a complex $K$ is \textit{minimal} if and only if it is isomorphic to $\mathcal{N}^2(K)$~\cite[Lemma 3.6]{StrongHomotopy}. This means that we can keep collapsing our complex $K$  by applying $\mathcal{N}^2(.)$ iteratively until we reach the core of the complex $K$. The sequence $K,\mathcal{N}^2(K),...,\mathcal{N}^{2p}(K)$ is a decreasing sequence in terms of number of simplices.

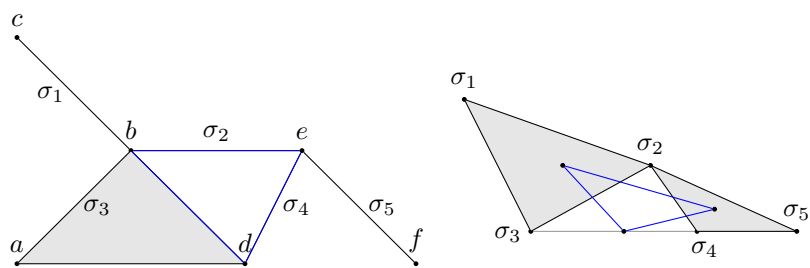
\begin{figure}
\begin{center}
\begin{tikzpicture}
[scale=3, vertices/.style={draw, fill=black, circle, inner sep=0.5pt}]
\node[vertices, label =$a$] (a) at (-0.5,-0.5) {};
\node[vertices, label =$b$] (b) at (0,0) {};
\node[vertices, label =$c$] (c) at (-0.5,0.5) {};
\node[vertices, label =$d$] (d) at (0.5,-0.5) {};
\node[vertices, label =$e$] (e) at (0.75,0) {};
\node[vertices, label =$f$] (f) at (1.25,-0.5) {};

\filldraw[fill=black!10, draw=black] (-0.5,-0.5)--(0,0) node[midway,right] {$\sigma_3$} --(0.5,-0.5)--cycle ;
\draw(b)--node[midway,left] {$\sigma_1$}(c);
\draw(b)--node[midway,above] {$\sigma_2$}(e);
\draw(d)--node[midway,right] {$\sigma_4$}(e);
\draw(e)--node[midway,right] {$\sigma_5$}(f);

\draw[draw =blue] (b)--(d);
\draw[draw =blue] (d)--(e);
\draw[draw =blue] (b)--(e);

\end{tikzpicture}
\begin{tikzpicture}
[scale=3.5, vertices/.style={draw, fill=black, circle, inner sep=0.5pt}]
\node[vertices, label = {[shift={(-0.3,-0.3)}]$\sigma_3$}] (f) at (0,-0.25) {};
\node[vertices, label = $\sigma_1$] (g) at (-0.25,0.25) {};
\node[vertices, label = $\sigma_2$] (h) at (0.45, 0) {};
\node[vertices, label = {[shift={(0.1,-0.5)}]$\sigma_4$}] (i) at (0.625,-0.25) {};
\node[vertices, label = $\sigma_5$] (j) at (1,-0.25) {};

\filldraw[fill=black!10, draw=black] (0,-0.25)--(-0.25,0.25)--(0.45, 0)--cycle ;
\filldraw[fill=black!10, draw=black] (0.45, 0)--(0.625,-0.25)--(1,-0.25)--cycle;
\draw [fill=black!10, draw=black, opacity=0.4](f)--(i);

\node[vertices] (k) at (0.12,0) {};
\node[vertices] (l) at (0.691, -0.166) {};
\node[vertices] (m) at (0.35,-0.25) {};

\draw [draw =blue](k)--(l);
\draw [draw =blue](l)--(m);
\draw [draw =blue](k)--(m);

\end{tikzpicture}
\end{center}
\caption{Left: $K$ (in grey), Right: $\mathcal{N}(K)$ (in grey) and $\mathcal{N}^2(K)$ (in \blue{blue}). $\mathcal{N}^2(K)$ is isomorphic to a full-subcomplex of $K$ highlighted in \blue{blue} on the left.}
\label{fig:threecomplexes}
\end{figure}

\section{Strong collapse of a simplicial complex}
\label{sec:algo}
In this section, we describe an algorithm to strong collapse a simplicial complex $K$, provide the details of the implementation and analyze its complexity. We construct $\mathcal{N}^2(K)$ as defined in Section \ref{sec:Nerve}.

\subparagraph{Data structure:}
Basically, we represent $K$ as the adjacency matrix $M$ between the vertices and the maximal simplices of $K$. We will simply call $M$ the adjacency matrix of $K$. The rows of $M$ represent the vertices and the columns represent the maximal simplices of $K$. For convenience, we will identify a row (resp. column) and the vertex (resp. maximal simplex) it represents. An entry $M[v_i][\sigma_j]$ associated with a vertex $v_i$ and a maximal simplex $\sigma_j$ is set to $1$ if $v_i \in \sigma_j$, and to $0$ otherwise. For example, the matrix $M$ in the left of the \autoref{tab:algoTable} corresponds to the leftmost simplicial complex $K$ in \autoref{fig:threecomplexes}. Usually, $M$ is very sparse. Indeed, each column contains at most $d+1$ non-zero elements since the simplices of a $d$-dimensional complex have at most $d+1$ vertices, and each line contains at most $\Gamma_0$ non-zero elements where $\Gamma_0$ is an upper bound on the number of maximal simplices incident to a given vertex. As already mentionned, in many practical situations, $\Gamma_0$ is a small fraction of the number of maximal simplices. It is therefore beneficial to store $M$ as a list of vertices and a list of maximal simplices. Each vertex $v$ in the list of vertices points to the maximal simplices that contain $v$, and each simplex in the list of maximal simplices points to  its vertices.  This data structure is similar to the SAL data structure of~\cite{gamma}.

\subparagraph{Core algorithm:}
Given the adjacency matrix $M$ of $K$, we compute the adjacency matrix $C$ of the \textit{core} $K^c$. It turns out that using basic row and column removal operations, we can easily compute $C$ from $M$. Loosely speaking our algorithm recursively computes $\mathcal{N}^2(K)$ until it reaches $K^c$. 

The columns of $M$ (which represent the maximal simplices of $K$) correspond to the vertices of $\mathcal{N}(K)$. Also, the columns of $M$ that have a non-zero value in a particular row $v$ correspond to the maximal simplices of $K$ that share the vertex associated with row $v$. Therefore, 
 each row  of $M$ represents a simplex of the nerve $\mathcal{N}(K)$. 
 Not all simplices of $\mathcal{N}(K)$ are associated with rows of $M$ but all maximal simplices are since they correspond to subsets of maximal simplices with a common vertex.
To remedy this situation, we \textit{remove} all the rows of $M$ that correspond to non-maximal simplices of $\mathcal{N}(K)$. This results in a new smaller matrix $M$ whose transpose, noted $\mathcal{N}(M)$,  is the adjacency matrix of the nerve $\mathcal{N}(K)$. 
We then exchange the roles of rows and columns (which is the same as taking the transpose) and run the very same procedure as before so as to obtain the adjacency matrix $\mathcal{N}^2(M)$ of $\mathcal{N}^2(K)$. 

The process is iterated as long as the matrix can be reduced.  Upon termination, we output the reduced matrix $C := \mathcal{N}^{2p}(M),$ for some $p \geq 1$, which is the adjacency matrix  of the core $K^c$ of $K$. Removing a row or column is the most basic operation of our algorithm. We will discuss it in more detail later in the paragraph \textit{Domination test}.

\begin{table}
\begin{center}
\scalebox{1.15}{
  \begin{tabular}{ | l | l | a | l |  l | l | }
    \hline
     $ $  & $\sigma_1$ & $\sigma_2$ & $\sigma_3$ & $\sigma_4$ & $\sigma_5$ \\ \hline
     a & 0 & 0 & 1 & 0 & 0 \\ \hline
	 b & 1 & 1 & 1 & 0 & 0 \\ \hline
     c & 1 & 0 & 0 & 0 & 0 \\ \hline
     d & 0 & 0 & 1 & 1 & 0 \\ \hline
	 \rowcolor{brown}      
     e & 0 & 1 & 0 & 1 & 1 \\ \hline
     f & 0 & 0 & 0 & 0 & 1 \\ 
    \hline
  \end{tabular}
\begin{tabular}{ | l | l | l | l |  l | l | }
    \hline
      $ $  		&	b 	&	d 	&	e 	\\ \hline
	 $\sigma_1$ 	&	1	&	0	&	0	\\ \hline
 	$\sigma_2$ 	&	1	&	0	&	1	\\ \hline
 	$\sigma_3$ 	&	1	&	1	&	0	\\ \hline
 	$\sigma_4$ 	&	0	&	1	&	1	\\ \hline
 	$\sigma_5$ 	&	0	&	0	&	1	\\ \hline
  \end{tabular}
\begin{tabular}{ | l | l | l |  l | }
    \hline
     $ $  & $\sigma_2$ & $\sigma_3$ & $\sigma_4$  \\ \hline
     b    & 1 & 1 & 0  \\ \hline
     d    & 0 & 1 & 1  \\ \hline
     e    & 1 & 0 & 1  \\ 
     \hline
     
  \end{tabular}}
\end{center}
\caption{From left to right $M$, $\mathcal{N}(M)$ and $\mathcal{N}^2(M)$} \label{tab:algoTable}
\end{table}

\subparagraph{Example:}  As mentioned before, the matrix $M$ in the left of the \autoref{tab:algoTable} represents the simplicial complex $K$ in the left of \autoref{fig:threecomplexes}. We go through the rows first, rows $a$ and $c$ are subsets of row $b$ and row $f$ is a subset of $e$. Removing rows $a$, $c$ and $f$ and transposing $M$ yields the adjacency matrix $\mathcal{N}(M)$ of $\mathcal{N}(K)$ in the middle. Now, row $\sigma_1$ is a subset of $\sigma_2$ and of $\sigma_3$, and $\sigma_5$ is a subset of $\sigma_2$ and of $\sigma_4$. We remove these two rows of $\mathcal{N}(M)$ and transpose $\mathcal{N}(M)$ so as to get $\mathcal{N}^2(M)$ (the rightmost matrix of \autoref{fig:threecomplexes}), which corresponds to the core drawn in blue in \autoref{fig:threecomplexes}. 

\subparagraph{Domination test: } \label{para:AnPntView} Now we explain in more detail how to detect  the rows that need to be removed.
Let $v$ be a row of $M$ and $\sigma_v$ be the associated simplex in $\mathcal{N}(K)$. 
If $\sigma_v$ is not a maximal simplex of $\mathcal{N}(K)$, it is a proper face of some maximal simplex $\sigma_{v'}$ of $\mathcal{N}(K)$. Equivalently, the row $v'$ of $M$ that is associated with $\sigma_{v'}$ contains row $v$ in the sense that the non zero elements of $v$ appear in the same columns as the non zero elements of $v'$. We will say that row $v$ is dominated by row $v'$ and determining if a row is dominated by another one will be called the row domination test.
Notice that when a row $v$ is dominated by a row $v'$, the same is true for the associated vertices since
all the maximal simplices that contain vertex $v$ also contain vertex $v'$, which  is the criterion to determine if $v$ is dominated by $v'$ (See Remark 1 in Section ~\ref{Rem:equDef}). The algorithm 
removes all dominated rows and therefore all dominated vertices of $K$. 

After removing rows,  the algorithm removes  the columns that  are no longer maximal in $K$, which might happen since we removed some rows.  
Removing a column may lead in turn to new dominated vertices and therefore new rows to be removed. 
When the algorithm stops, there are no rows to be removed and we have obtained the core $K^c$ of the complex $K$. Note that the algorithm provides a constructive proof of \autoref{Thm:NerveSquare}. 

Removing columns is done in very much the same way:  we just exchange the roles of rows and columns. 


\subparagraph{Computing the retraction map \textbf{$r$}:} 
The algorithm  also provides a direct way to compute the retraction map $r$ defined in Section ~\ref{map:retmap}. 
The retraction map corresponding to the strong collapses executed by the algorithm can be constructed as follows. A row $r$ being removed in $M$ corresponds to a dominated vertex in $K$ and the row which \textit{contains} $r$ corresponds to a dominating vertex. Therefore we map the dominated vertex to the dominating vertex 
and compose all such maps to get  the final retraction map from $K$ to its core $K^c$. The final map is simplicial as well, as it is a composition of simplicial maps.

\subparagraph{Reducing the number of domination tests:} 
We first observe that, when one wants to determine if a row $v$ is dominated by some other row, we don't need to test $v$ with all other rows but with at most $d$ of them. Indeed, at most $d+1$ rows can intersect a given column since a simplex can have at most $d+1$ vertices.
For example, in \autoref{tab:algoTable} (Left), to check if  row e (highlighted in \textcolor{brown}{brown}) is dominated by another row, we pick the first non-zero column $\sigma_2$ (highlighted in Gray) and compare e with the non-zero entries \{b\} of $\sigma_2$. 

A second observation is that we don't need to test all rows and columns for domination, but only the so-called candidate rows and columns.
We define a row $r$ to be a \textbf{candidate row} for the next iteration if at least one column containing one of the non-zero elements of $r$ has been removed in the previous column removal iteration. Similarly, by exchanging the roles of rows and columns, we define the \textbf{candidate columns}. Candidate rows and columns are the only rows or columns that need to be considered in the \textit{domination} tests of the algorithm.  Indeed, a column $\tau$ of $M$ whose non-zero elements all belong to rows that are present from the previous \textit{iteration} cannot be dominated by another column $\tau'$ of $M$, since $\tau$ was not dominated at the previous iteration and no new non-zero elements have {ever been added by the algorithm}. 
The same argument follows for the candidate rows. 
 
We maintain two \textit{queues}, one for the candidate columns (colQueue) and one for the candidate rows (rowQueue). These queues are implemented as First in First out (FIFO) queues. 
At each iteration, we \textit{pop out}  a candidate row or column from its respective queue and test whether it is dominated or not. After each successful domination test, we \textit{push} the candidate columns or rows in their appropriate queue in preparation for the subsequent iteration. In the first iteration, we \textit{push} all the rows in rowQueue and then alternatively use  colQueue and rowQueue. Algorithm \autoref{alg:core} gives the pseudo code of our algorithm.

\begin{algorithm}
    \caption{Core algorithm}
    \label{alg:core}
    \begin{algorithmic}[1] 
        \Procedure{Core}{$M$} \Comment{ Returns the matrix corresponding to the core of $K$}
			\State $rowQueue \gets \text{ \textit{push} all rows of M (all vertices of K)}$
			\State $colQueue \gets \text{empty}$         
            \While{ rowQueue is not empty} 
				\State $v \gets \textit{pop}(rowQueue)$ 					
					\State $\sigma \gets \text{the first non-zero column of } v $
					\For{non-zero rows $w$ in $\sigma$}
						\If{ $v$ is a subset of $w$}
							\State \text{Remove $v$ from $M$}
							\State \text{\textit{push} all non-zero columns $\tau$ of $v$ to $colQueue$ if not pushed before} 
							\State \text{break}
						\EndIf
					\EndFor  
			\EndWhile
			\While{ colQueue is not empty}
				\State $\tau \gets \textit{pop}(colQueue)$ 					
				\State $v \gets \text{the first non-zero row of } \tau $	
					\For{non-zero columns $\sigma$ in $v$}
						\If{ $\tau$ is subset of $\sigma$}
							\State \text{Remove $\tau$ from $M$}
							\State \text{\textit{push} all non-zero rows $w$ of $\tau$ to $rowQueue$ if not pushed before} 
							\State \text{break}
						\EndIf
					\EndFor            	
            \EndWhile\label{coreendwhile}
            \If{ rowQueue is not empty}
            		\State \text{ GOTO $4$} 
            	\EndIf	
            \State \textbf{return} $M$\Comment{The core consists of the remaining rows and columns}
        \EndProcedure
    \end{algorithmic}
\end{algorithm}

\subparagraph{Time Complexity: } The most basic operation in our algorithm is to determine if a row is dominated by another given row, and similarly for columns. In our implementation,  the rows (columns) of the matrix that are considered by the algorithm are stored as sorted lists.
Checking if one sorted list is a subset of another sorted list can be done  in time $\mathcal{O}(l)$, where $l$ is  the size of the longer list.
Note that the length of a row list is at most $\Gamma_0$ where $\Gamma_0$ denotes an upper bound on the number of maximal simplices incident to a vertex. The length of a column list is at most $d+1$ where $d$ is the dimension of the complex. Hence
checking if a row is dominated by another row takes $\mathcal{O}(\Gamma_0)$ time and
checking if a column is dominated by another column takes $\mathcal{O}(d)$ time.

At each iteration on the rows (Lines 7-13 of Algorithm ~\ref{alg:core}), each row is checked against at most $d$ other rows (since a maximal simplex has at most $d+1$ vertices), and at each iteration of the columns  (Lines 18-24 of Algorithm ~\ref{alg:core}), each column is checked against at most $\Gamma_0$ other columns (since a vertex can belong to at most $\Gamma_0$ maximal simplices). Since, at each iteration on the rows, we  remove at least one row, the total number of iterations on the rows is at most $O(v^2)$, where $v$ is the total number of vertices of the complex $K$. Similarly,  at each iteration on the columns, we remove at least one column and the total number of iterations on columns is $O(m^2)$,  where $m$ is the total number of maximal simplices of the complex $K$. The worst-case time complexity of our algorithm is therefore  $\mathcal{O}(v^2\Gamma_0d + m^2\Gamma_0d)$.
{In practice, $m$ is much smaller than $n$, the total number of simplices,   and $\Gamma_0$ is much smaller than $\Gamma$, the {maximum } number  of simplices incident on a vertex.   Typically $\Gamma$ grows exponentially with $d$ while $\Gamma_0$ remains almost constant as $d$ increases. See Table 5 in~~\cite{gamma} and related results in~~\cite{gamma2}, and also the plots in Section~\ref{sec:experiments}.}

\section{Strong collapse of a sequence of simplicial complexes } \label{SCZG}
In this section, we will present our main result that the persistence homology of a sequence of simplicial complexes is preserved under strong collapse.
To be able to present our main result, we need to begin with some brief background on zigzag persistence. Readers interested in more details can refer to \cite{CarlssonZigzag,MorozovZigzag,quivers}. 

A \textit{sequence of simplicial complexes} $\mathcal{Z}$ : $\{K_1 \xrightarrow{f_1} K_2 \xleftarrow{g_2} K_3 \xrightarrow{f_3} \cdots \xrightarrow{f_{(n-1)}} K_n\}$ is a sequence of complexes $K_i$s connected through simplicial maps  $\xrightarrow{f_i}$s and $\xleftarrow{g_j}$s. In the most general case, the maps are in both directions ${\rightarrow, \leftarrow}$ and the sequence is called a \textbf{zigzag sequence}.  If all the maps are forward,  i.e. it consists only of the $f_i$s, the sequence is called a simplicial \textbf{tower}.  If all the maps are inclusions, we have a sequence of nested simplicial complexes  called a \textbf{filtration}. Our results apply to all types of sequences and will be described for zigzag sequences.

Once we compute the homology classes of all $K_i$s, we get the sequence $\mathcal{P}(\mathcal{Z})$ : $\{H_p(K_1) \xrightarrow{f_1^*} H_p(K_2) \xleftarrow{g_2^*} H_p(K_3) \xrightarrow{f_3^*} \cdots \xrightarrow{f_{(n-1)}^*} H_p(K_n)\}$. Here $H_p(-)$ denotes the homology class of dimension $p$ with coefficients from a field $\mathbb{F}$ and $*$ denotes an induced homomorphism. $\mathcal{P}(\mathcal{Z})$ is a sequence of vector spaces connected through homomorphisms, called a \textbf{zigzag module}. More formally, a \textit{zigzag module} $\mathbb{V}$ is a sequence of vector spaces $\{V_1 \xrightarrow{} V_2 \xleftarrow{} V_3 \xrightarrow{} \cdots \xleftrightarrow{} V_n\}$ connected with homomorphisms $\{\xrightarrow{}, \xleftarrow{} \}$ between them. A  zigzag module arising from a sequence of simplicial complexes captures the evolution of the topology of the sequence. 

For two integers $b$ and $d$, $1 \leq b \leq d \leq n$; we can define an \textbf{interval module} $\mathbb{I}[b,d]$ by assigning $V_i$ to $\mathbb{F}$ when $i \in [b,d]$, and null spaces otherwise, the maps between any two $\mathbb{F}$ vector spaces is identity and is zero otherwise. For example $\mathbb{I}[2,4] : \{0 \xrightarrow{0} \mathbb{F} \xleftarrow{I} \mathbb{F} \xrightarrow{I} \mathbb{F} \xleftarrow{0}  0 \xrightarrow{0}  0\}$, here $n = 6$. Any zigzag module can be \textit{decomposed} as the direct sum of \textit{finitely} many interval modules, which is unique upto the permutations of the interval modules ~\cite{CarlssonZigzag}. The multiset of all the intervals $[b_j, d_j]$ corresponding to the interval module decomposition of any zigzag module is called a \textbf{zigzag (persistence) diagram}. The zigzag diagram completely characterizes the zigzag module, that is, there is bijective correspondence between them \cite{CarlssonZigzag, CarlssonZomorodian}.

Two different zigzag modules $\mathbb{V} : \{V_1 \xrightarrow{} V_2 \xleftarrow{} V_3 \xrightarrow{} \cdots \xleftrightarrow{} V_n\}$ and $\mathbb{W} : \{W_1 \xrightarrow{} W_2 \xleftarrow{} W_3 \xrightarrow{} \cdots \xleftrightarrow{} W_n\}$, connected through a set of homomorphisms $\phi_i: V_i \rightarrow W_i$ are \textbf{equivalent} if the $\phi_i$s are isomorphisms and the following diagram commutes ~\cite{CarlssonZigzag,quivers}.
\begin{center}
 \begin{tikzcd}
V_1 \arrow{r}{} \arrow{d}{\phi_1} & V_2 \arrow{d}{\phi_2} & V_3 \arrow{l}{} \arrow{d}{\phi_3} &
 \cdots &  \arrow{r} & V_{n-1} \arrow{r}{} \arrow{d}{\phi_{n-1}} & V_n \arrow{d}{\phi_n} \\
 W_1 \arrow{r}{} & W_2 & W_3 \arrow{l}{} & \cdots & \arrow{r} & W_{n-1} \arrow{r}{} & W_n
 \end{tikzcd}
 \end{center}
Note that the \textit{length} of the modules and the directions of the arrows in them should be consistent. Two equivalent zigzag modules will have the same interval decomposition, therefore the same zigzag diagram.

\subparagraph{Strong collapse of the zigzag module:} Given a zigzag sequence $\mathcal{Z}$ : $\{K_1 \xrightarrow{f_1} K_2 \xleftarrow{g_2} K_3 \xrightarrow{f_3} \cdots \xrightarrow{f_{(n-1)}} K_n\}$. We define the \textbf{core sequence} $\mathcal{Z}^c$ of $\mathcal{Z}$ as $\mathcal{Z}^c$ : $\{K_1^c \xrightarrow{f_1^c} K_2^c \xleftarrow{g_2^c} K_3^c \xrightarrow{f_3^c} \cdots \xrightarrow{f_{(n-1)}^c} K_n^c\}$. Where $K_i^c$ is the core of $K_i$. The forward maps are defined as, $f_j^c := r_{j+1}f_ji_j$; and the backward maps are defined as $g_j^c := r_{j}g_ji_{j+1}$. The maps $ i_j : K_j^c \hookrightarrow K_j $  and $ r_j : K_j \rightarrow K_j^c$ are the composed inclusions and the retractions maps defined in Section ~\ref{map:retmap} respectively. We call the procedure of forming the core sequence using the cores and the induced simplicial maps as \textbf{core-assembly}.

\begin{theorem} \label{thm:equivalent_module}
Zigzag modules $\mathcal{P}(\mathcal{Z})$ and $\mathcal{P}(\mathcal{Z}^c)$ are equivalent.
\end{theorem}
\begin{proof} Consider the following diagram
\begin{center}
\begin{tikzcd}
K_1 \arrow{r}{f_1} \arrow{d}{r_1} & K_2 \arrow{d}{r_2} & K_3 \arrow{l}{g_2} \arrow{d}{r_3} & 
 \cdots &  \arrow{r} & K_{n-1} \arrow{r}{f_{n-1}} \arrow{d}{r_{n-1}} & K_n \arrow{d}{r_n} \\
 K_1^c \arrow{r}{f_1^c} & K_2^c & K_3^c \arrow{l}{g_2^c} & \cdots & \arrow{r} & K_{n-1}^c \arrow{r}{f_{n-1}^c} & K_n^c
 \end{tikzcd}
\end{center}
and the associated  diagram after computing the $p$-th homology groups
 \begin{center}
 \begin{tikzcd}
 H_p(K_1) \arrow{r}{f_1^*} \arrow{d}{r_1^*} & H_p(K_2) \arrow{d}{r_2^*} & H_p(K_3) \arrow{l}{g_2^*} \arrow{d}{r_3^*} &
 \cdots &  \arrow{r} & H_p(K_{n-1}) \arrow{r}{f_{n-1}^*} \arrow{d}{r_{n-1}^*} & H_p(K_n) \arrow{d}{r_n^*} \\
 H_p(K_1^c) \arrow{r}{(f_1^c)^*} & H_p(K_2^c) & H_p(K_3^c) \arrow{l}{(g_2^c)^*} & \cdots & \arrow{r} & H_p(K_{n-1}^c) \arrow{r}{(f_{n-1}^c)^*} & H_p(K_n^c)
 \end{tikzcd}
 \end{center}
Since there exists a strong deformation retract between $r_j$ and $i_j$, the induced homomorphisms $r_j^*$ and $i_j^*$ are isomorphisms \cite[Corollary 2.11]{Hatcher}. Also, $f_j^cr_j = r_{j+1}f_ji_jr_j$ is contiguous to $r_{j+1}f_j$, since $i_jr_j$ is contiguous to the identity on $K_j$ and contiguity is preserved under composition, see ~~\cite[Proposition 2.9]{StrongHomotopy} and similarly $g_j^cr_{j+1}$ is contiguous to $ r_jg_j$. Now, since contiguous maps are homotopic at the level of geometric realization and homotopic maps induce the same homomorphism, we have $(f_j^cr_j)^* = (r_{j+1}f_j)^*$ and thus $(f_j^c)^*r_j^* = r_{j+1}^*f_j^*  $ and similarly $(g_j^c)^*r_{j+1}^* = r_j^*g_j^*$, see ~\cite[Proposition (1) page 111]{Hatcher}. Therefore all the squares in the lower diagram commute and the set of maps $r_j^*$s are isomorphisms, therefore $\mathcal{P}(\mathcal{Z})$ and $\mathcal{P}(\mathcal{Z}^c)$ are \textit{equivalent} and hence their zigzag diagrams are identical.
\end{proof}

\noindent{\bf Remark.} The above result can be extended to  multidimensional persistence using the more general notion of quiver representation ~\cite{quivers}.

\subparagraph{Approximation of the persistence diagram:}  The complexes in the sequence $\mathcal{Z}$ are usually associated to different real values of a scale parameter. 
For example, in  a Rips-Vietoris (RV) filtration, the filtration value of a simplex is the length of the longest edge of the simplex and the filtration consists of a  sequence of elementary inclusions (i.e. inclusion of a single simplex) with increasing \textit{filtration values}. A persistent pair then associates 
the filtration value of a simplex that creates a homology cycle with the filtration value of the simplex that kills the cycle. The \textit{length} of the persistence pair is the difference between the two filtration values.

One could choose to strong collapse the complexes after each such inclusion and, as per Theorem \ref{thm:equivalent_module}, the persistence diagram of the core sequence $\mathcal{Z}^c$ would be the same as of $\mathcal{Z}$. 
%
%
A more efficient usage of our method is to strong collapse the complexes less often, i.e. after several inclusions rather than just one. This will result in a faster algorithm but comes with a cost: the computed PD is only approximate. 
We call \textbf{snapshots}  the values of the scale parameter at which we choose to strong collapse the complex. The difference between two consecutive snapshots is called a \textbf{step}. We approximate the \textbf{filtration value} of a simplex as the value of the snapshot at which it first appears. It is not hard to see that our algorithm will report all persistence pairs that are separated by at least one snapshot. Hence if all steps are equal to some $\epsilon>0$, we will compute all the persistence pairs  whose lengths are at least $\epsilon$. It follows that the  bottleneck distance between  the computed PD and the exact one is at most $\epsilon$.


\section{Computational experiments}
\label{sec:experiments}
In this section, we present some computational experiments to showcase the efficiency achieved by our approach. For our experiments, we choose two synthetic datasets and four real datasets. The algorithms to strong collapse a simplicial complex (Algorithm~\ref{alg:core}) and to form the core sequence (core-assembly) have been coded in C++ and will be available as an open-source package of a next release of the Gudhi library~\cite{Gudhi}. The code has been compiled using the compiler `clang-900.0.38' and all computations were performed on a `2.8 GHz Intel Core i5' machine with 16 GB of available RAM.

The experiments in this paper are limited to RV filtrations, by far the most commonly used type of sequences in Topological Data Analysis.  We intend to experiment on Zigzag sequences in a companion paper.

For each data set,  
we select a number of snapshots and \textit{independently} strong collapse all  the complexes associated to these snapshots. We  then assemble the resulting individual \textit{cores} using the induced simplicial maps introduced in Section~\ref{SCZG}. The resulting core sequence with induced simplicial maps between the collapsed complexes is in general a simplicial tower we call the {\em core tower}. We then convert the core tower into an equivalent filtration using the Sophia software ~\cite{Sophia}, { which implements the algorithm described by Kerber and Schreiber in ~\cite{HannaTower}}. Finally, we run the persistence algorithm of the Gudhi library ~\cite{Gudhi} to obtain the persistence diagram (PD) of the {equivalent} filtration. 
The total time to compute the PD of the core sequence is the sum of three terms: 1. the \textit{maximum} time taken to collapse all the individual complexes {(assuming they are computed in parallel)}, 2. the time taken to assemble the individual \textit{cores} to form the core tower, 3. the time to compute the persistent diagram of the core tower. \autoref{tab:filtCollapse} summarises the results of the experiments. In both cases, the original filtration and the core tower, we use Gudhi through Sophia using the command \text{<./sophia -cgudhi inputTowerFile outputPDFile>}. When we use the -cgudhi option, Sophia reports two computation times. The first one is the total time taken by Sophia which includes (1) reading the tower, (2) transforming it to a filtration and (3) computing PD using Gudhi. The second reported time is just the time taken by Gudhi to compute PD. In our comparisons, we just report the time taken by Gudhi for the original filtration,  while, for the core tower, we report the total time taken by Sophia.

\begin{table}
\begin{center}
\scalebox{1.15}{
  \begin{tabular}{| l | l | l | l | l | l |  l | l | l | l |}
    \hline
    $\mathcal{X}$ 	&	$1$-sphere  	&	$2$-Annulus 	&	dragon  & netw-sc   & senate  	&	eleg    	\\ \hline
 		Snp 	&	80	&	80	&	46	&	69	&	107	&	77	\\ \hline
 		Flt($10^6$) 	&	\red{0.12}	&	\red{13.91}	&	\red{7.96}	&	\red{22.35}	&	\red{2.56}	&	\red{1.18}	\\ \hline
 		Twr 	&	54	&	252	&	1,641	&	380	&	104	&	298	\\ \hline
 		EqF 	&	\blue{573}	&	\blue{1,954}	&	\blue{8,437} 	&	\blue{957}	&	\blue{270}	&	\blue{431}	\\ \hline  
		Flt/EqF($10^3$)	&	0.21	&	7.12	&	0.94	&	23.35	&	9.48	&	2.74	\\ \hline \hline		
 		PDF 	&	 \red{0.65} 	&	 \red{174.18} 	&	 \red{69.92} 	&	 \red{243.86} 	&	 \red{24.92}  	&	 \red{10.87}  	\\ \hline
 		MCT 	&	0.005	&	0.022	&	0.065	&	0.009	&	0.003	&	0.002	\\ \hline
 		AT 	&	0.045	&	0.136	&	0.408	&	0.078	&	0.06	&	0.157	\\ \hline
 		PDT 	&	0.01 	&	0.02 	&	0.08 	&	0.01 	&	0.005	&	0.006	\\ \hline
 		Total 	&	 \blue{0.060}   &	 \blue{0.178} 	&	 \blue{0.553} &	 \blue{0.097} 	&	  \blue{0.068} 	&	 \blue{0.165}   	\\ \hline  	
 		PDF/Total	&	10.8	&	978.5	&	126.4	&	2514.0	&	366.5	&	65.9	\\ \hline	
  \end{tabular}}
 \end{center}
\caption{The rows are, from top to bottom: dataset $\mathcal{X}$, number of snapshots (snp), total number of simplices in the original filtration (Flt) in millions, number of simplices in the collapsed tower (Twr), total number of simplices in the equivalent filtration (EqF), ratio of Flt and EqF (Flt/EqF) in thousands, PD computation time for the original filration (PDF), maximum collapse time (MCT), assembly time (AT), PD computation time of the tower (PDT), sum MCT+AT+PDT (Total), ratio of PDF and Total (PDF/Total). All times are noted in seconds. For the first three datasets, we sampled points randomly from the initial datasets and averaged the results over five trials.} \label{tab:filtCollapse}
\end{table} 
The dataset of the first column ($1$-sphere) of \autoref{tab:filtCollapse} consists of  $100$ random points sampled from a unit circle in dimension 2. The dataset of the second column ($2$-Annulus) consists of $150$ random points sampled from a two dimension annulus of radii $\{0.6,1\}$. For all the other experiments,  we use datasets from a publicly available repository \cite{NinaData}. These datasets have been previously used to benchmark different publicly available software computing PH~\cite{NinaPaper}. For the third experiment (\textit{dragon}), we randomly picked  $150$ points from the $2000$ points of the dataset \textbf{drag 2} of \cite{NinaData}. The fourth and fifth column respectively correspond to the dataset \textbf{netw-sc} and \textbf{senate} of \cite{NinaData}, here we used the distance matrix. The sixth column corresponds to the dataset \textbf{eleg} of \cite{NinaData}, and here again we used the distance matrix. The first three datasets are point sets in Euclidean space. For the other three, the distance matrices of the datasets  were available at $\cite{NinaData}$. The [initial value, step, final value] of the scale parameter are  $[0.1, 0.005, 0.5]$, $[0.1, 0.005, 0.5]$, $[0, 0.001, 0.046]$, $[0.1, 0.05, 3.5]$, $[0, 0.001, 0.107]$ and $[0, 0.001, 0.077 ]$ for the examples in Table \ref{tab:filtCollapse} (from left to right). The choices of these parameters are simplistic but reasonable: the final value ensures that the size of the original filtration is not too large and we can run experiments on our machine; steps are constant and small enough to keep the bottleneck distance small.  
For more detail about the datasets and the computation of the distance matrices of the last three datasets please refer to \cite{NinaPaper}. 

Noticeably, in our experiments, the computing time of our approach is reduced by 1 to 3 orders of magnitude, and  the gain increases with the size of the filtration. A similar reduction of 2 to 4 orders of magnitude is achieved for the number of simplices. In a separate set of experiments we have reported the total time not just maximum time, see Table ~\ref{tab:PH-Collapser}. In Table ~\ref{tab:PH-Collapser}, we have also mentioned the time to compute the RV complex corresponding to the largest scale parameter which contributes the most in the maximum time. One can observe from Table ~\ref{tab:PH-Collapser} that the total time is not directly proportional to the number of snapshots used. In fact, in most of the cases the total time stays within a factor of two of Rips-Comp-Time (maximum time). The only exception is the dataset \textit{senate}, which has comparatively fewer points: the time to compute the RV complex is relatively fast. However, here too the ratio is less than eight, which is much less than 403, the number of snapshots. This clearly indicates that in our approach the time to perform the strong collapses and the core assemblies  together with the time to compute the PD is much smaller than the time to compute the RV complex. This implies that one can refine the PD (through smaller steps) at a very low cost. 

The plots below count  the maximal simplices and the dimensions of the complexes across the filtration (in solid) and the collapsed tower (as dashed). \textcolor{blue}{Blue} and \textcolor{red}{red} correspond  respectively to the filtration and the collapsed tower of the data \textbf{netw-sc}. Similarly \textcolor{green}{green} and \textcolor{brown}{brown} correspond  respectively to the filtration and the collapsed tower of the data \textbf{senate}. Finally, \textbf{black} and \textcolor{cyan}{cyan} correspond to the filtration and the collapsed tower of the data \textbf{eleg} respectively. We can observe 
that in all cases the number of maximal simplices never increases. Also they are far fewer in number compared to the total number of simplices. Observe that for the uncollapsed filtrations \textcolor{blue}{blue}, \textcolor{green}{green} and \textbf{black}, the dimension of the complexes increases quite rapidly with the snapshot index. {Another key fact to observe is that the dimension of the complexes in the corresponding core tower are much smaller than their counterparts in the filtration. This has a huge effect on the performances since the  total number of simplices depends exponentially on the dimension.\\

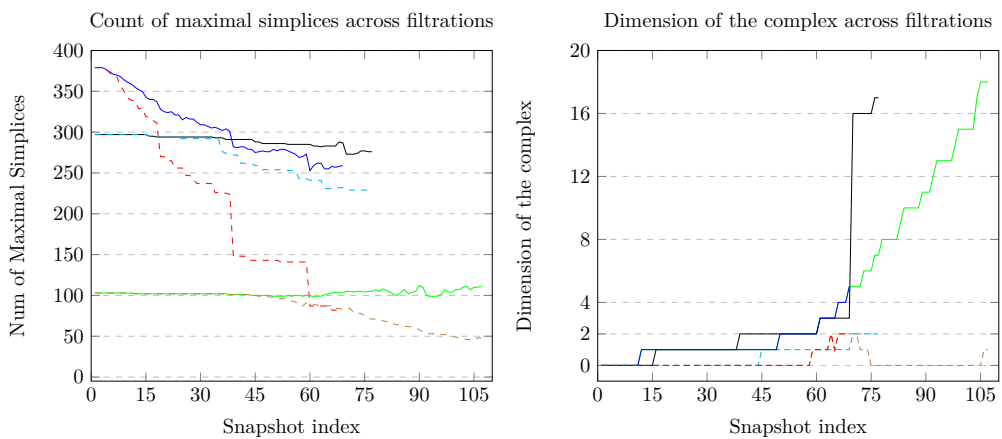
\begin{tikzpicture}[scale = 0.77]\label{tab:maxPlot}
\begin{axis}[
    title={Count of maximal simplices across filtrations},
    xlabel={Snapshot index},
    ylabel={Num of Maximal Simplices},
    xmin=0, xmax=110,
    ymin=-5, ymax=400,
    xtick={0, 15, 30, 45, 60, 75, 90, 105, 120},
    ytick={0, 50, 100, 150, 200, 250, 300, 350, 400},
    legend pos=north west,
    ymajorgrids=true,
    grid style=dashed,
]
 
\addplot [color=blue,
    ]
    coordinates { (1, 379) (2, 379) (3, 379) (4, 377) (5, 374) (6, 371) (7, 370) (8, 368) (9, 364) (10, 361) (11, 359) (12, 355) (13, 352) (14, 350) (15, 343) (16, 340) (17, 340) (18, 337) (19, 328) (20, 325) (21, 324) (22, 325) (23, 321) (24, 321) (25, 315) (26, 318) (27, 316) (28, 316) (29, 312) (30, 309) (31, 309) (32, 307) (33, 305) (34, 306) (35, 304) (36, 302) (37, 304) (38, 301) (39, 282) (40, 281) (41, 282) (42, 282) (43, 279) (44, 279) (45, 275) (46, 276) (47, 277) (48, 276) (49, 276) (50, 279) (51, 277) (52, 279) (53, 279) (54, 277) (55, 275) (56, 273) (57, 269) (58, 270) (59, 273) (60, 253) (61, 259) (62, 262) (63, 262) (64, 255) (65, 255) (66, 258) (67, 257) (68, 258) (69, 259)
     };
\addplot
	[color = red, dashed]   
 coordinates {(1, 379) (2, 379) (3, 379) (4, 377) (5, 373) (6, 369) (7, 368) (8, 357) (9, 350) (10, 342) (11, 339) (12, 338) (13, 329) (14, 329) (15, 319) (16, 318) (17, 314) (18, 311) (19, 270) (20, 270) (21, 265) (22, 265) (23, 256) (24, 256) (25, 256) (26, 247) (27, 247) (28, 242) (29, 237) (30, 237) (31, 237) (32, 237) (33, 237) (34, 226) (35, 226) (36, 226) (37, 225) (38, 224) (39, 148) (40, 148) (41, 148) (42, 148) (43, 144) (44, 143) (45, 143) (46, 143) (47, 143) (48, 143) (49, 143) (50, 143) (51, 143) (52, 141) (53, 141) (54, 141) (55, 141) (56, 141) (57, 141) (58, 141) (59, 141) (60, 87) (61, 87) (62, 87) (63, 87) (64, 87) (65, 87) (66, 82) (67, 82) (68, 82) (69, 82)
 };
\addplot [color=green,]
coordinates {(1, 103) (2, 103) (3, 103) (4, 103) (5, 103) (6, 103) (7, 103) (8, 103) (9, 103) (10, 103) (11, 103) (12, 102) (13, 102) (14, 102) (15, 102) (16, 102) (17, 102) (18, 102) (19, 102) (20, 102) (21, 102) (22, 102) (23, 102) (24, 102) (25, 102) (26, 102) (27, 102) (28, 102) (29, 102) (30, 102) (31, 102) (32, 102) (33, 102) (34, 102) (35, 102) (36, 102) (37, 102) (38, 101) (39, 101) (40, 101) (41, 101) (42, 101) (43, 101) (44, 100) (45, 100) (46, 100) (47, 100) (48, 100) (49, 100) (50, 99) (51, 99) (52, 99) (53, 99) (54, 100) (55, 99) (56, 100) (57, 100) (58, 99) (59, 100) (60, 100) (61, 98) (62, 99) (63, 99) (64, 101) (65, 102) (66, 103) (67, 104) (68, 103) (69, 103) (70, 105) (71, 105) (72, 104) (73, 105) (74, 105) (75, 104) (76, 105) (77, 105) (78, 107) (79, 105) (80, 108) (81, 106) (82, 104) (83, 101) (84, 100) (85, 104) (86, 107) (87, 102) (88, 104) (89, 105) (90, 112) (91, 110) (92, 100) (93, 98) (94, 99) (95, 99) (96, 101) (97, 107) (98, 104) (99, 104) (100, 107) (101, 107) (102, 110) (103, 111) (104, 107) (105, 110) (106, 110) (107, 111)
 }; 
\addplot [color = brown, dashed]
coordinates {(1, 103) (2, 103) (3, 103) (4, 103) (5, 103) (6, 103) (7, 103) (8, 103) (9, 103) (10, 103) (11, 103) (12, 102) (13, 102) (14, 102) (15, 102) (16, 102) (17, 102) (18, 102) (19, 102) (20, 102) (21, 102) (22, 102) (23, 102) (24, 102) (25, 102) (26, 102) (27, 102) (28, 102) (29, 102) (30, 102) (31, 102) (32, 102) (33, 102) (34, 102) (35, 102) (36, 102) (37, 102) (38, 101) (39, 101) (40, 101) (41, 101) (42, 101) (43, 101) (44, 100) (45, 100) (46, 99) (47, 99) (48, 99) (49, 98) (50, 98) (51, 96) (52, 96) (53, 96) (54, 95) (55, 93) (56, 92) (57, 90) (58, 87) (59, 91) (60, 91) (61, 89) (62, 88) (63, 88) (64, 88) (65, 87) (66, 88) (67, 84) (68, 84) (69, 83) (70, 84) (71, 84) (72, 79) (73, 78) (74, 78) (75, 75) (76, 72) (77, 71) (78, 71) (79, 70) (80, 69) (81, 68) (82, 65) (83, 65) (84, 64) (85, 63) (86, 62) (87, 62) (88, 61) (89, 60) (90, 58) (91, 55) (92, 53) (93, 53) (94, 53) (95, 53) (96, 52) (97, 52) (98, 51) (99, 50) (100, 49) (101, 48) (102, 47) (103, 46) (104, 46) (105, 45) (106, 48) (107, 48) 
};
\addplot[color = black,]
coordinates{ (1, 297) (2, 297) (3, 297) (4, 297) (5, 297) (6, 297) (7, 297) (8, 297) (9, 297) (10, 297) (11, 297) (12, 297) (13, 297) (14, 297) (15, 297) (16, 295) (17, 295) (18, 294) (19, 294) (20, 294) (21, 294) (22, 294) (23, 294) (24, 294) (25, 294) (26, 294) (27, 294) (28, 294) (29, 294) (30, 294) (31, 294) (32, 294) (33, 294) (34, 293) (35, 293) (36, 293) (37, 291) (38, 291) (39, 291) (40, 291) (41, 291) (42, 291) (43, 291) (44, 291) (45, 288) (46, 288) (47, 286) (48, 286) (49, 286) (50, 286) (51, 286) (52, 286) (53, 286) (54, 285) (55, 285) (56, 285) (57, 285) (58, 285) (59, 285) (60, 285) (61, 283) (62, 283) (63, 282) (64, 283) (65, 283) (66, 283) (67, 283) (68, 288) (69, 287) (70, 273) (71, 273) (72, 273) (73, 274) (74, 277) (75, 277) (76, 276) (77, 276) 
}; 	
\addplot[color = cyan, dashed]
coordinates{(1, 297) (2, 297) (3, 297) (4, 297) (5, 297) (6, 297) (7, 297) (8, 297) (9, 297) (10, 297) (11, 297) (12, 297) (13, 297) (14, 297) (15, 297) (16, 295) (17, 295) (18, 294) (19, 294) (20, 294) (21, 294) (22, 294) (23, 294) (24, 292) (25, 292) (26, 292) (27, 292) (28, 292) (29, 292) (30, 292) (31, 292) (32, 292) (33, 292) (34, 291) (35, 290) (36, 276) (37, 274) (38, 274) (39, 272) (40, 272) (41, 262) (42, 262) (43, 261) (44, 261) (45, 259) (46, 259) (47, 254) (48, 254) (49, 254) (50, 254) (51, 254) (52, 254) (53, 254) (54, 253) (55, 253) (56, 253) (57, 243) (58, 243) (59, 243) (60, 241) (61, 241) (62, 241) (63, 241) (64, 231) (65, 231) (66, 231) (67, 231) (68, 232) (69, 232) (70, 229) (71, 229) (72, 229) (73, 230) (74, 229) (75, 229) (76, 229) (77, 229)
};
\end{axis}
\end{tikzpicture}
\begin{tikzpicture}[scale = 0.77]\label{tab:dimPlot}
\begin{axis}[
    title={Dimension of the complex across filtrations},
    xlabel={Snapshot index},
    ylabel={Dimension of the complex},
    xmin=0, xmax=110,
    ymin=-1, ymax=20,
    xtick={0, 15, 30, 45, 60, 75, 90, 105, 120},
    ytick={-3,0,2,4,8,12,16,20},
    legend pos=north west,
    ymajorgrids=true,
    grid style=dashed,
]
 
\addplot [color=green,]
coordinates {(1, 0) (2, 0) (3, 0) (4, 0) (5, 0) (6, 0) (7, 0) (8, 0) (9, 0) (10, 0) (11, 0) (12, 1) (13, 1) (14, 1) (15, 1) (16, 1) (17, 1) (18, 1) (19, 1) (20, 1) (21, 1) (22, 1) (23, 1) (24, 1) (25, 1) (26, 1) (27, 1) (28, 1) (29, 1) (30, 1) (31, 1) (32, 1) (33, 1) (34, 1) (35, 1) (36, 1) (37, 1) (38, 1) (39, 1) (40, 1) (41, 1) (42, 1) (43, 1) (44, 1) (45, 1) (46, 1) (47, 1) (48, 1) (49, 1) (50, 2) (51, 2) (52, 2) (53, 2) (54, 2) (55, 2) (56, 2) (57, 2) (58, 2) (59, 2) (60, 2) (61, 3) (62, 3) (63, 3) (64, 3) (65, 3) (66, 4) (67, 4) (68, 4) (69, 5) (70, 5) (71, 5) (72, 5) (73, 6) (74, 6) (75, 6) (76, 7) (77, 7) (78, 8) (79, 8) (80, 8) (81, 8) (82, 8) (83, 9) (84, 10) (85, 10) (86, 10) (87, 10) (88, 10) (89, 11) (90, 11) (91, 11) (92, 12) (93, 13) (94, 13) (95, 13) (96, 13) (97, 13) (98, 14) (99, 15) (100, 15) (101, 15) (102, 15) (103, 15) (104, 17) (105, 18) (106, 18) (107, 18)
 }; 
\addplot [color = brown, densely dashed]
coordinates {(1, 0) (2, 0) (3, 0) (4, 0) (5, 0) (6, 0) (7, 0) (8, 0) (9, 0) (10, 0) (11, 0) (12, 0) (13, 0) (14, 0) (15, 0) (16, 0) (17, 0) (18, 0) (19, 0) (20, 0) (21, 0) (22, 0) (23, 0) (24, 0) (25, 0) (26, 0) (27, 0) (28, 0) (29, 0) (30, 0) (31, 0) (32, 0) (33, 0) (34, 0) (35, 0) (36, 0) (37, 0) (38, 0) (39, 0) (40, 0) (41, 0) (42, 0) (43, 0) (44, 0) (45, 0) (46, 0) (47, 0) (48, 0) (49, 0) (50, 0) (51, 0) (52, 0) (53, 0) (54, 0) (55, 0) (56, 0) (57, 0) (58, 0) (59, 1) (60, 1) (61, 1) (62, 1) (63, 1) (64, 2) (65, 1) (66, 2) (67, 2) (68, 2) (69, 2) (70, 2) (71, 2) (72, 1) (73, 1) (74, 1) (75, 0) (76, 0) (77, 0) (78, 0) (79, 0) (80, 0) (81, 0) (82, 0) (83, 0) (84, 0) (85, 0) (86, 0) (87, 0) (88, 0) (89, 0) (90, 0) (91, 0) (92, 0) (93, 0) (94, 0) (95, 0) (96, 0) (97, 0) (98, 0) (99, 0) (100, 0) (101, 0) (102, 0) (103, 0) (104, 0) (105, 0) (106, 1) (107, 1) 
};
\addplot[color = black,]
coordinates{ (1, 0) (2, 0) (3, 0) (4, 0) (5, 0) (6, 0) (7, 0) (8, 0) (9, 0) (10, 0) (11, 0) (12, 0) (13, 0) (14, 0) (15, 0) (16, 1) (17, 1) (18, 1) (19, 1) (20, 1) (21, 1) (22, 1) (23, 1) (24, 1) (25, 1) (26, 1) (27, 1) (28, 1) (29, 1) (30, 1) (31, 1) (32, 1) (33, 1) (34, 1) (35, 1) (36, 1) (37, 1) (38, 1) (39, 2) (40, 2) (41, 2) (42, 2) (43, 2) (44, 2) (45, 2) (46, 2) (47, 2) (48, 2) (49, 2) (50, 2) (51, 2) (52, 2) (53, 2) (54, 2) (55, 2) (56, 2) (57, 2) (58, 2) (59, 2) (60, 2) (61, 3) (62, 3) (63, 3) (64, 3) (65, 3) (66, 3) (67, 3) (68, 3) (69, 3) (70, 16) (71, 16) (72, 16) (73, 16) (74, 16) (75, 16) (76, 17) (77, 17)
}; 	
\addplot[color = cyan, densely dashed]
coordinates{(1, 0) (2, 0) (3, 0) (4, 0) (5, 0) (6, 0) (7, 0) (8, 0) (9, 0) (10, 0) (11, 0) (12, 0) (13, 0) (14, 0) (15, 0) (16, 0) (17, 0) (18, 0) (19, 0) (20, 0) (21, 0) (22, 0) (23, 0) (24, 0) (25, 0) (26, 0) (27, 0) (28, 0) (29, 0) (30, 0) (31, 0) (32, 0) (33, 0) (34, 0) (35, 0) (36, 0) (37, 0) (38, 0) (39, 0) (40, 0) (41, 0) (42, 0) (43, 0) (44, 0) (45, 1) (46, 1) (47, 1) (48, 1) (49, 1) (50, 1) (51, 1) (52, 1) (53, 1) (54, 1) (55, 1) (56, 1) (57, 1) (58, 1) (59, 1) (60, 1) (61, 1) (62, 1) (63, 1) (64, 1) (65, 1) (66, 1) (67, 1) (68, 1) (69, 1) (70, 2) (71, 2) (72, 2) (73, 2) (74, 2) (75, 2) (76, 2) (77, 2)
};
\addplot [color= blue, ]
  coordinates { (1, 0) (2, 0) (3, 0) (4, 0) (5, 0) (6, 0) (7, 0) (8, 0) (9, 0) (10, 0) (11, 0) (12, 1) (13, 1) (14, 1) (15, 1) (16, 1) (17, 1) (18, 1) (19, 1) (20, 1) (21, 1) (22, 1) (23, 1) (24, 1) (25, 1) (26, 1) (27, 1) (28, 1) (29, 1) (30, 1) (31, 1) (32, 1) (33, 1) (34, 1) (35, 1) (36, 1) (37, 1) (38, 1) (39, 1) (40, 1) (41, 1) (42, 1) (43, 1) (44, 1) (45, 1) (46, 1) (47, 1) (48, 1) (49, 1) (50, 2) (51, 2) (52, 2) (53, 2) (54, 2) (55, 2) (56, 2) (57, 2) (58, 2) (59, 2) (60, 2) (61, 3) (62, 3) (63, 3) (64, 3) (65, 3) (66, 4) (67, 4) (68, 4) (69, 5) };
\addplot	[color = red, densely dashed]   
 coordinates {(1, 0) (2, 0) (3, 0) (4, 0) (5, 0) (6, 0) (7, 0) (8, 0) (9, 0) (10, 0) (11, 0) (12, 0) (13, 0) (14, 0) (15, 0) (16, 0) (17, 0) (18, 0) (19, 0) (20, 0) (21, 0) (22, 0) (23, 0) (24, 0) (25, 0) (26, 0) (27, 0) (28, 0) (29, 0) (30, 0) (31, 0) (32, 0) (33, 0) (34, 0) (35, 0) (36, 0) (37, 0) (38, 0) (39, 0) (40, 0) (41, 0) (42, 0) (43, 0) (44, 0) (45, 0) (46, 0) (47, 0) (48, 0) (49, 0) (50, 0) (51, 0) (52, 0) (53, 0) (54, 0) (55, 0) (56, 0) (57, 0) (58, 0) (59, 1) (60, 1) (61, 1) (62, 1) (63, 1) (64, 2) (65, 1) (66, 2) (67, 2) (68, 2) (69, 2)
 };

\end{axis}
\end{tikzpicture}

Observations from the plots combined with the experimental results of Table ~\ref{tab:filtCollapse} clearly indicate that our method is extremely fast and memory efficient.

\subparagraph{Comparison with Ripser:}
In the above experiments the comparison was between computation time to compute the persistence diagram (PD), i.e. we didn't consider the time taken to compute the RV complex in both cases. Also, we used Gudhi to compute the PD of both the original filtration and the collapsed equivalent filtration. Now we present some experimental results comparing our approach with Ripser \cite{ripser}, which is the state of the art software to compute the PD of RV filtrations. We again used Gudhi to compute the PD of the collapsed sequence but here for a fair comparison we include the time taken to compute the RV filtration. We call our package to preprocess the initial sequence and construct the collapsed sequence as the \textit{PH-Collapser}. The comparison is done on the three datasets \textbf{netw-sc}, \textbf{senate} and \textbf{eleg} from \cite{NinaData}. The reported time is the total time taken by Gudhi, which includes the time taken to compute the entire RV filtration (at snapshot values) and the time taken to collapse all the subcomplexes and assemble their cores and to transform them into an equivalent filtration and then finally to compute the PD of the equivalent filtration. 

\begin{table}[]
\begin{center}
 \begin{tabular}{|c|c|c|c|c|c|c|c|}
\hline
\multirow{2}{*}{Data} & \multirow{2}{*}{Pnt} & \multirow{2}{*}{Threshold} & %
    \multicolumn{5}{c|}{PH-Collapser(Gudhi)} \\
\cline{4-8}
 & & & Dim & Rips-Comp-Time & Total-Time & Steps & TotSnaps \\
\hline
netw-sc & 379 & 4.5 & \blue{41} & \blue{13s} & \blue{21s} & 0.02 & 213\\ 
\hline
'' & '' & 5.5 & \blue{57}&  \blue{117s} & \blue{144s}  & 0.02 & 263 \\ 
\hline
senate & 103 & 0.415 &  \blue{54} & \blue{1.7s} & \blue{13.1s} & 0.001 & 403\\ 
\hline
eleg & 297 & 0.3 & \blue{105} & \blue{443s} & \blue{578.3s} & 0.001 & 284 \\ 

\hline
\end{tabular}
 \end{center}
\caption{The columns are, from left to right: dataset (Data), number of points (Pnt), maximum scale parameter (Threshold), dimension of the RV Complex (Dim), time taken to compute the RV complex (Rips-Comp-Time), total time taken by PH-Collapser (Gudhi) (Total-Time), incremental steps of subcomplexes (Steps) and total number of snapshots used (TotSnaps). All times are averaged over five trials except the last row (2 times).} \label{tab:PH-Collapser}
\end{table} 

Command  \text{<./ripser inputData --format distances --threshold inputTh --dim inputDim >} was used to run Ripser and we used the distance matrix format for all the datasets. Differently from PH-Collapser, Ripser needs a parameter --dim until which it computes the PD. For a given \textit{threshold} (maximum scale parameter) and \textit{dim}, Ripser basically computes the \textit{dim}-skeleton of the RV complex and then computes the PD of the skeleton. For the given threshold, we compute the complete RV complex until its full dimension. However, we only use the maximal faces to represent the complex which again saves a lot of memory and time. This allows us to compute the PD in all the dimensions, i.e until the dimension of the RV complex.

Table ~\ref{tab:PH-Collapser} contains the results of the experiments done using PH-Collapser and Table ~\ref{tab:Ripser} contains the results of Ripser. By comparing the tables, PH-Collapser clearly outperforms Ripser by a huge margin considering that we compute the persistence diagram until the full dimension. Ripser performs quite well for computing PD in low-dimensions, however as we move to intermediate dimensions it slows down quite considerably and in some cases (dimension above 7) the size of the complex is so huge that Ripser crashed due to memory overload. In Table ~\ref{tab:Ripser}, we provide the running time of Ripser with increasing dimension of the PD computed.

\begin{table}[h]
\begin{center}
 \begin{tabular}{|c|c|c|c|c|c|c|c|c|}
\hline
\multirow{2}{*}{Data} & \multirow{2}{*}{Pnt} & \multirow{2}{*}{Threshold}  
 & \multicolumn{2}{c|}{Val} & \multicolumn{2}{c|}{Val} & \multicolumn{2}{c|}{Val} \\
\cline{4-9}
  &  &  & Dim & Time & Dim & Time & Dim & Time \\
\hline
netw-sc & 379 & 4.5 & 4 & 3.8s & \red{5} & \red{21.5s} & \red{7} & \red{357s} \\ 
\hline
'' & '' & 5.5 & 4 & 25.3s & \red{5} & \red{231.2s} & \red{6} & \red{$\infty$} \\ 
\hline
senate & 103 & 0.415 & 3 & 0.52s & 4 & 5.9s & \red{5} & \red{52.3s}   \\ 
\hline
'' & '' & '' & \red{6} & \red{406.8s} & \red{7} & \red{$\infty$} & &   \\ 
\hline
eleg & 297 & 0.3 & 3 & 8.9s & 4 & 217s & \red{5} & \red{$\infty$} \\ 
\hline
\end{tabular}
 \end{center}
\caption{The columns are, from left to right: dataset (Data), number of points (Pnt), maximum scale parameter (Threshold), input dimension for Ripser (Dim), total time taken by Ripser (Time). Most results are averaged over five trials except the longer ones. $\infty$ in the Time column means that the experiment ran longer than 12hrs or crashed due to memory overload.} \label{tab:Ripser}
\end{table} 


As mentioned before, we define the filtration value of a simplex as the value of the snapshot parameter at which it appears for the first time, whereas in the case of Ripser it is the length of the longest edge (1-simplex) it contains. Therefore, the computed PD by PH-Collapser is not exactly the same as the one computed by Ripser, see Section~\ref{SCZG} (Approximation of the persistence diagram). However, in the above experiments, we choose steps that are very small so that the bottleneck distance between the two PD returned by Ripser and PH-collapser for a given data set is also very small. 

Note that the choice of \textit{snapshots} is arbitrary and could be done in a non uniform way after analyzing the distribution of the length of the edges in the Rips-complex at a relatively small increase in the total computing time.

\section{Discussion}
In this article, we presented a novel approach to compute the persistence homology of a sequence of simplicial complexes. Our approach is based on a technique called strong collapse that has been introduced by Barmak and Minian \cite{StrongHomotopy}. It works very well in pratice and, as shown using  publicly available data,  is extremely fast and memory efficient.  
%
%
%
We believe that the solid mathematical foundations presented in \cite{StrongHomotopy}, its applicability to all kinds of sequences of simplicial complexes,  and the availability of the simple and efficient algorithms developed in this article, strong collapses will be immensely useful  to reduce the complexity of many problems in computational topology. 

On the theoretical side,  this work raises several questions. In particular, it would be nice to have  theoretical guarantees on the amount of reduction the algorithm can achieve. We intend to explore this and related issues in future work.

\subparagraph*{Acknowledgements.}

We want to thank Marc Glisse for his useful discussions, Mathijs Wintraecken for reviewing a draft of the article. We also want to thank Francois Godi and Siargey Kachanovich for their help with Gudhi and Hannah Schreiber for her help with Sophia and Ulrich Bauer for his help with Ripser.

\nocite{*}
\bibliography{mybib}{}
\end{document}